\documentclass[pra,aps,showpacs,amssymb,amsfonts,superscriptaddress,twocolumn,nofootinbib,10pt]{revtex4-2}
\usepackage[T1]{fontenc}
\usepackage{bm,bbm}
\usepackage{epstopdf}
\usepackage{times}
\usepackage{graphicx,amsthm,epstopdf,amsmath}
\usepackage{subfigure,color}
\usepackage[draft]{fixme}
\usepackage{diagbox}
\usepackage[super]{nth}
\usepackage[top=1.25in, bottom=1.25in, left=0.70in, right=0.70in]{geometry}
\definecolor{myurlcolor}{rgb}{0,0,0.7}
\definecolor{myrefcolor}{rgb}{0.8,0,0}
\usepackage{hyperref}
\hypersetup{colorlinks, linkcolor=myrefcolor,
citecolor=myurlcolor, urlcolor=myurlcolor}
\usepackage{hyperref}
\usepackage{mathtools}

\newcommand{\beq}[0]{\begin{equation}}
\newcommand{\eeq}[0]{\end{equation}}
\newcommand{\bw}[0]{\begin{widetext}}
\newcommand{\ew}[0]{\end{widetext}}
\newcommand{\bc}[0]{\begin{center}}
\newcommand{\ec}[0]{\end{center}}
\newcommand{\bwn}[0]{\begin{widetext}\begin{eqnarray}}
\newcommand{\ewn}[0]{\end{eqnarray}\end{widetext}}
\newcommand{\beqn}[0]{\begin{eqnarray}}
\newcommand{\eeqn}[0]{\end{eqnarray}}

\newcommand{\ket}[1]{|#1\rangle}

\newcommand{\non}[0]{\nonumber\\}

\newcommand{\spann}[0]{\mathrm{span}}

\def\calB{{\cal B}}

\def\calH{{\cal H}}

\def\calS{{\cal S}}

\newtheorem{thm}{Theorem}

\newtheorem{defi}[thm]{Definition}

\newtheorem{prop}[thm]{Proposition}
\newtheorem{fakt}[thm]{Fact}

\newcommand{\beu}{\begin{equation}}
\newcommand{\eeu}{\end{equation}}

\newcommand{\be}{\begin{eqnarray}}
\newcommand{\ee}{\end{eqnarray}}

\newcommand{\ba}{\begin{array}}
\newcommand{\ea}{\end{array}}

\newcommand{\cee}[1]{\mathbb{C}^{#1}}

\newcommand{\dmin}[0]{d_{\mathrm{min}}}
\newcommand{\dmax}[0]{d_{\mathrm{max}}}

\begin{document}

\title{
Negative result about the construction of genuinely entangled subspaces from unextendible product bases}
\author{Maciej Demianowicz}
 \affiliation{{\it\small Institute of Physics and Applied Computer Science, Faculty of Applied Physics and Mathematics,
		Gda\'nsk University of Technology, Narutowicza 11/12, 80-233 Gda\'nsk, Poland}}
\email{maciej.demianowicz@pg.edu.pl}

\begin{abstract}
	
Unextendible product bases (UPBs) provide a versatile tool with various applications across different areas of quantum information theory. Their comprehensive characterization is thus of great importance and has been a subject of vital interest for over two decades now.   An open question asks about the existence of UPBs, which are genuinely unextendible, i.e., they are not extendible even with biproduct vectors. In other words, the problem  is to verify whether there exist genuinely entangled subspaces (GESs), subspaces composed solely of genuinely multiparty entangled states, complementary to UPBs. We solve this problem in the negative for many sizes of UPBs in different multipartite scenarios. In particular, in the all-important case of equal local dimensions, we show that there are always forbidden cardinalities for such UPBs, including the minimal ones corresponding to GESs of the maximal dimensions.

\end{abstract}
\maketitle

\section{Introduction}
Unextendible product bases (UPBs), that is,  incomplete sets of orthogonal product vectors, such that no other product vector orthogonal to their span exists, are a very important notion of quantum information theory \cite{upb,big-upb,upb-comment}. While being significant from a purely mathematical standpoint, 
UPBs are also related to a variety of problems with more practical implications. For example, they provide examples of sets of product vectors which cannot be distinguished with local operations and classical communication, the so-called nonlocality without entanglement phenomenon \cite{without-entanglement}.
Also, by the very definition, they give rise to constructions of completely entangled subspaces (CESs) \cite{Parthasarathy,Bhat,WalgateScott}, that is, subspaces void of product vectors. Importantly, mixed states supported on such subspaces, which are necessarily entangled, have positive partial transposition (PPT) across any partition \cite{upb} and thus are indistillable or, in other words, bound entangled (BE) \cite{PH,H3}.

Much effort has thus been put into finding and characterizing UPBs in different bipartite and multipartite setups (a largely incomplete list of works includes, e.g., \cite{alon-lovasz,feng, djokovic,johnston-qubits,min-upb-johnston,3x4,large-upb}).
However, all those attempts were not concerned about the type of entanglement of entangled states in the orthocomplement of a UPB. 
In \cite{upb-to-ges} we  thus posed a problem whether it is possible to construct multipartite  genuinely unextendible product bases (GUPBs), that is UPBs  which are unextendible in a stronger sense, even with a biproduct vector (i.e., a vector product across a bipartition). This amounts to asking about the construction of  genuinely entangled subspaces (GESs), i.e., subspaces composed only of genuinely multiparty entangled (GME) states \cite{upb-to-ges,cubitt} (see also \cite{Parthasarathy}), from UPBs. 
While there has been progress in constructing GESs \cite{Agrawal, ent-ges,approach,stabilizerGES,antipin,antipin-diagrammatic}, with a very recent result solving the problem in full generality \cite{universal-ges}, the problem of determining the (non)existence of GUPBs has remained basically unexplored  (albeit see \cite{4x4}) and open.

 We address the problem in the present work and  derive a lower bound on the cardinalities of GUPBs in a general multipartite scenario. The bound is universally applicable in systems with equal local dimensions, in which case it  shows that  for many cardinalities of UPBs, including the minimal permissible ones, they are always extendible with biproduct vectors, i.e., GUPBs do not exist in those cases. In turn,  GESs with many dimensionalities, including the maximal ones, cannot be constructed from UPBs. 
  Our result, besides being a contribution to the theory of UPBs, has two important immediate implications. It rules out the possibility of  constructing through the usual approach PPT GME BE states  with certain ranks  from UPBs and shows that in many cases strongly nonlocal sets of product vectors cannot be built from, natural candidates for this task, GUPBs   (cf.  \cite{strong-nonlocality,yuan,strong-upb, strong-upb-hetero}).

\section{Preliminaries }

We provide here a terse summary of the necessary notions and results.

We consider $n$-partite Hilbert spaces with local dimensions $d_i$, $\calH=\bigotimes_{i=1}^n \cee{d_i}$, and  assume $n \ge 3$  and $d_i \ge 3$ as these cases only are relevant. We will write $(d_1,d_2,\dots,d_n)$ in  nondecreasing order to denote the set of local dimensions and its permutations, e.g., $(3,3,4)$ means that any two of the subsystems are qutrits and one of them is a ququad. Local subsystems are denoted by $A_i$ and the whole system is $\textbf{A}$. A state $\ket{\psi}_{\textbf{A}}\in \calH$ is called {\it fully product}, or simply {\it product}, if it is a product of single-party states, $\ket{\psi}_{\textbf{A}}=\ket{\varphi}_{A_1}\otimes\cdots \otimes \ket{\xi}_{A_n}$. If it is not the case, a state is said to be {\it entangled}.
Some entangled states are {\it biproduct}, that is they are product across at least one of the bipartitions of the parties $S|\bar{S}$, $ S\cup \bar{S}=\textbf{A}$,  i.e., $\ket{\psi}_{\textbf{A}}=\ket{\zeta}_{S}\otimes\ket{\xi}_{\bar{S}}$. If an entangled state is {\it not} biproduct we say that it is {\it genuinely multiparty entangled (GME)}.

 Subspaces composed only of entangled vectors are called completely entangled (CES) \cite{Parthasarathy,Bhat,WalgateScott}; among those there are subspaces only with GME vectors and they are named  genuinely entangled (GES) \cite{upb-to-ges,Agrawal,ent-ges,approach,antipin,condition,universal-ges,antipin-diagrammatic,cubitt}. The maximal dimension of a GES is \cite{cubitt,upb-to-ges}
 $D - D/\dmin -\dmin+1$,  with $D=\Pi_{i=1}^n d_i$ and $\dmin=\min (d_1,d_2,\dots,d_n)$.

 An unextendible product basis (UPB) is an incomplete  set of product vectors with the property that there does not exist a product vector orthogonal to all of them \cite{upb,big-upb}. Obviously, a UPB defines a CES in its orthocomplement.
 The simplest example of a UPB is the following four-element set of vectors from $(\cee{2})^{\otimes 3}$: $\calS=\{ \ket{0}\ket{0}\ket{0}, \ket{1}\ket{+}\ket{-}, \ket{-}\ket{1}\ket{+}, \ket{+}\ket{-}\ket{1}\}$, $\ket{\pm}=\ket{0}\pm \ket{1}$. This UPB is known under the name SHIFTS \cite{upb}. A particularly useful tool in the analyses of UPBs (or, more generally, sets of product vectors) is an orthogonality graph (see Fig. \ref{ort-graf}).
 
 \begin{figure}[h!]
 	\includegraphics[width=4.6cm,height=4.6cm]{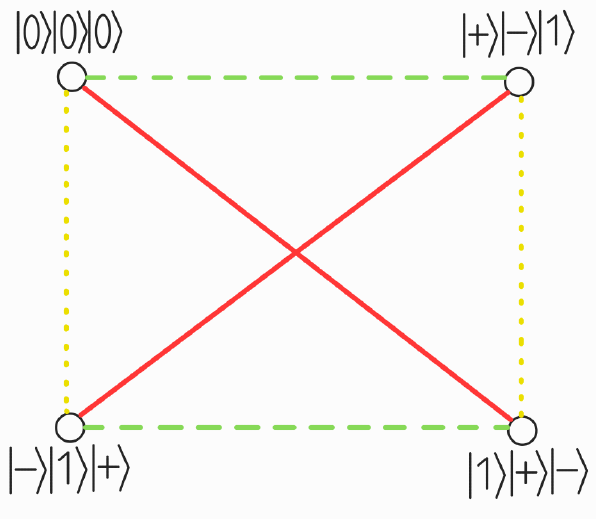}
 	\caption{A set of product vectors, in particular, a UPB, can be represented as a complete graph with colored edges with colors corresponding to sites on which two states are orthogonal. The figure shows such a graph for the SHIFTS UPB $\calS$. Red solid lines corresponds to orthogonality on the first site, gold dotted lines the second, and green dashed lines the third.
 	}\label{ort-graf}
 \end{figure}
 
 In the present paper we ask about the (non)existence of UPBs, which are not extendible in a stronger sense, namely, there does not exist a biproduct vector orthogonal to its elements ($\calS$ is clearly not such a set). Such UPBs must also be bipartite UPBs across any cut of the parties \cite{upb-to-ges}.  We propose the following
 
 \begin{defi}
 	A multipartite  UPB, which is a UPB  for any bipartition, is called a genuinely unextendible product basis (GUPB).
 \end{defi}
 
 Clearly, the orthocomplement of the span of a GUPB is a GES. The existence of GUPBs would provide then a very convenient way of constructing GESs and in turn GME mixed states.
 
It is known that UPBs do not exist in systems  $\cee{2}\otimes \cee{m}$, implying that GUPBs do not exist when at least one of the subsystems is a qubit, which is why in what follows we assume  $d_i \ge 3$. It is also well established that UPBs do not exist with arbitrary cardinalities. The following summarizes what is currently known regarding their sizes \cite{alon-lovasz,feng,no-very-large-upb,large-upb,min-upb-johnston}.
\begin{fakt}\label{min-bound}
(i) The theoretical  minimum number of elements in a UPB in  $\cee{d_1}\otimes \cee{d_2}$ ($d_1,d_2 \ge 3$)  is (a) $d_1+d_2$, when $d_1,d_2 \ge 4$ are even, or (b) $d_1+d_2-1$, in the remaining cases.
(ii) The achievable upper bound on the cardinality of a UPB is $d_1d_2-4$.
	\end{fakt}
By considering all bipartite cuts of an $n$-partite system with $d_i$-dimensional subsystems, it follows that the minimal permissible cardinality of a GUPB is 
$D/\dmin + \dmin$, when both $\dmin$ and $D/\dmin$ are even, and $D/\dmin + \dmin-1$ in the remaining cases.
  This already implies that maximal  GESs cannot be constructed from UPBs for  even $\dmin$ and $D/\dmin$  and arbitrary $n$; the maximal allowed dimensions of GESs 
   and the minimum cardinalities of UPBs simply do not match. Moreover, GESs of dimensions $1$, $2$, and $3$ cannot be constructed in this way either.
We will be further interested in eliminating cardinalities, which are theoretically possible by Fact \ref{min-bound}. 

Our proof is based on the necessary and sufficient condition for extendibility of a set of product vectors due to DiVincenzo and co-workers \cite{upb,big-upb}. We recall it here in the bipartite formulation since, as noted above, this is relevant for our purpose.

\begin{fakt}\label{crucial}
	Let $\calB=\{\ket{\varphi_i}\otimes \ket{\psi_i}\}_{i=1}^{k}$, $k\ge m+n-1$, be a set of product vectors from $\cee{m}\otimes \cee{n}$. There exists a product vector orthogonal to all elements of $\calB$, i.e., $\calB$ is extendible, if and only if there exist a partition of $\calB$ into disjoint sets $\calB_1$ and $\calB_2$ such that the $\ket{\varphi_i}$'s of $\calB_1$ do not span $\cee{m}$ and the $\ket{\psi_i}$'s of $\calB_2$ do not span $\cee{n}$.
\end{fakt}

\section{Main result}

Let us now  prove the main result of the present work, which is a general bound on the cardinalities of GUPBs, implying, in particular, that in systems with equal local dimensions, GESs of maximal dimensions (and many smaller ones) cannot be constructed from UPBs. 
	\begin{prop} \label{main-result}
		There do not exist genuinely unextendible product bases (GUPBs) with cardinalities $k$ satisfying
		\begin{equation}\label{bound}\displaystyle
		k \le D/\dmax+ \Big\lfloor \frac{D/\dmax-2}{n-1} \Big\rfloor,
		\end{equation}
		where $D=\Pi_{i=1}^n d_i$ and $\dmax=\max (d_1,d_2,,\dots,d_n)$.
	%
	%
	\newline\indent
The condition (\ref{bound}) is a non-trivial bound on the cardinality of a GUPB if and only if
\begin{equation}\label{iff-condition}
(n-1)\dmax < n\dmin,
\end{equation}
$\dmin=\min(d_1,d_2,\dots,d_n)$, and
\beqn\label{}
&&(d_1,d_2,d_3) \ne (2p,2p,3p-1),\label{parzyste}\\
&&(d_1,d_2,d_3) \ne(2p-1,\tilde{d},3p-2), \label{nieparzyste}
\eeqn
with$p=2,3,\dots $ and $2p-1 \le \tilde{d} \le 3p-2$. 
		In particular, in the case of equal local dimensions $d$, all UPBs with cardinalities $d^{n-1}+d-1$ and $d^{n-1}+d$ are not GUPBs. 
		In consequence, genuinely entangled subsapces (GESs) of those  maximal theoretically possible with the approach dimensions cannot be constructed from GUPBs.
		\end{prop}

\begin{proof}
Let $\calB=\{\ket{v_i}_{\textbf{A}}\}_{i=1}^k$, $\ket{v_i}_{\textbf{A}}=\bigotimes_{m=1}^n \ket{u_m^{(i)}}_{A_m}$, be a set of $k$ mutually orthogonal product vectors. 
Since orthogonality of the vectors from $\calB$  stems from orthogonality of local vectors $\ket{u_m^{(i)}}$ on different sites,
 for any $\ket{v_i}$, by the pigeonhole principle, there exist at least
 \begin{equation}\label{}\displaystyle
 s:=\Big\lceil \frac{k-1}{n} \Big\rceil
 \end{equation}
vectors orthogonal to this vector on one of the sites. Consider for simplicity one of the vectors, say $\ket{v_1}$. Further, let the $s$ vectors mentioned above, which are orthogonal to $\ket{v_1}$, be $\calB_1=\{\ket{v_2},\dots,\ket{v_{s+1}}\}$ and the corresponding site of orthogonality (this can be any site when all $d_i$'s are equal) be $A_p$ (see Fig. \ref{obrazek-dowod}). It follows that the vectors from $\calB_1$ do not span locally on $A_p$ the whole Hilbert space, i.e.,  $\dim\spann\:\{\ket{u_p^{(2)}},\dots,\ket{u_p^{(s+1)}}\} < d_p$. Now, if the remaining vectors $\calB_2:=\{\ket{v_1},\ket{v_{s+2}},\dots,\ket{v_k}\}$ do not span locally on $\textbf{A}\setminus A_p$ the whole  Hilbert space,
 then, by Fact \ref{crucial}, there exists a biproduct vector orthogonal to all $\ket{v_i}$'s from $\calB$, which is given as $\ket{u_{p}^{(1)}}\otimes \ket{\xi}$ with an $(n-1)$-partite vector $\ket{\xi} \perp \spann\:\calB_2^{\textbf{A}\setminus A_p}$, where $\calB_2^{\textbf{A}\setminus A_p}$ is the set of local vectors on $\textbf{A}\setminus A_p$ of the set $\calB_2$.
	A sufficient condition for the local rank deficiency of $\calB_2$ on  $\textbf{A}\setminus A_p$  is  simply that the number of states is smaller than the dimension of the local Hilbert space
	\begin{equation}\label{warunek}
	k-s =k-\Big\lceil \frac{k-1}{n} \Big\rceil \le D/d_p-1 :=w,
	\end{equation}
where $D=\Pi_{i=1}^nd_i$.
The function $f(k)=k-\lceil \frac{k-1}{n} \rceil$ is nondecreasing in $k$ [if $\frac{k-1}{n}$ is an integer, then $f(k+1)=f(k)$]
and thus we look for the largest $k$ satisfying (\ref{warunek}).  With this aim consider the equation
\begin{equation}\label{prosty}
k-s=w.
\end{equation}
It holds that
\begin{equation}\label{}
 \frac{k-1}{n}  \le s \le  \frac{k-1}{n} + \frac{n-1}{n}.
\end{equation}
Plugging $k$ from (\ref{prosty}) in the above, we obtain
\begin{equation}\label{}
\frac{w-1}{n-1} \le s \le \frac{w-1}{n-1}+1.
\end{equation}
It follows that the optimal value is $s = \lfloor \frac{w-1}{n-1}  \rfloor+1$. Inserting this into (\ref{prosty}), we obtain the value of the largest $k$ for which (\ref{warunek}) is satisfied,
\begin{equation}\label{}
w+ \Big \lfloor \frac{w-1}{n-1} \Big \rfloor+1,
\end{equation}
and in turn the  condition
\begin{equation}\label{weaker-condition}
k \le D/d_p+ \Big\lfloor \frac{D/d_p-2}{n-1} \Big\rfloor.
\end{equation}
Since we want to obtain a general bound,
we need to consider the least preferable situation, which occurs for $d_p=\dmax$, from which Eq. (\ref{bound}) follows.

The proof of the remaining two statements is in  the Appendix.

\end{proof}

\begin{figure}[h!]
	\includegraphics[width=6cm,height=6cm]{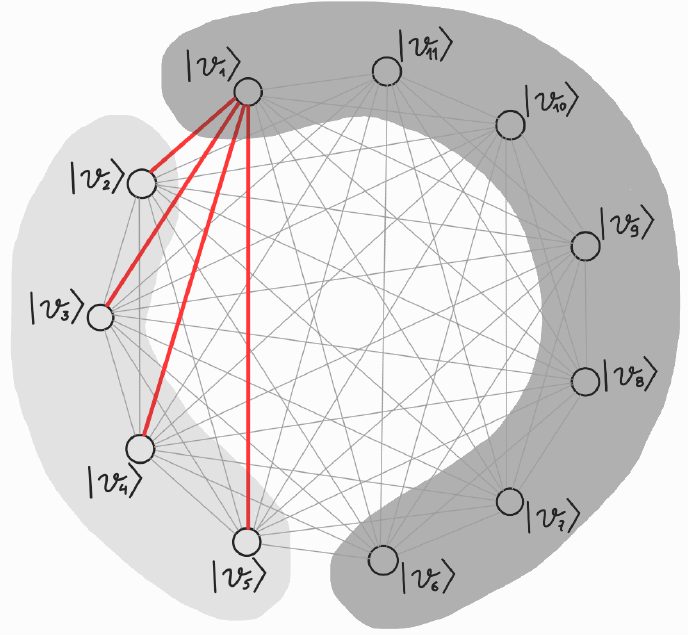}
	\caption{
		Orthogonality graph for an $11$-element set of mutually orthogonal product vectors in the three-qutrit case. It follows that vector $v_1$ must be orthogonal to at least four vectors on one of the sites, say $A_1$ (red thick edges; all the remaining irrelevant edges are light gray). This implies that vectors $v_2,v_3,v_4,v_5$ (light gray shaded  area) span locally on $A_1$ at most two-dimensional subspace. There remain only seven vectors ( dark gray shaded area), which means that they do not span locally on $A_2A_3$ the whole nine-dimensional Hilbert space. In turn, by Fact \ref{crucial}, there exists a biproduct vector orthogonal to $v_i$'s.  Since $11$ is the minimal theoretically allowed cardinality of a GUPB in this case, our result shows that the maximal GESs cannot be constructed from UPBs. The same argument holds for a $12$-element set, while it fails for $13$ elements, as there will be nine vectors in the dark gray shaded  area.
		}\label{obrazek-dowod}
\end{figure}

Table \ref{tabelka-wymiarow} shows which cardinalities apart from the minimal theoretical ones are excluded by Proposition \ref{main-result} in the $n$-qudit case for different values of $d$ and $n$.

\begin{center}
	\begin{table}
		\begin{tabular}{cccccc}
			\hline\hline \vspace{0.1cm}
			\backslashbox{$n$}{$d$}
			&3&4&5
			&6\\  \hline 
		\vspace{0.1cm}	3 & $[ 11,12 ]$  & $[ 20,23 ]$  & $[ 29,36]$  & $[ 42,53]$     \\ \vspace{0.1cm}
			4 & $[ 29,35 ]$  & $[ 68,84 ]$  & $[ 129,166]$  & $[ 222,287 ]$   \\
			5 & $[ 83,100]$  & $[ 260,319 ]$  & $[ 629,780 ]$  & $[ 1302,1619 ]$ \vspace{0.1cm} \\  \hline\hline
		\end{tabular}
		\caption{Cardinalities of GUPBs excluded by Proposition \ref{main-result} in the case of equal local dimensions $d$. The left endpoints of all intervals are minimal permissible cardinalities of GUPBs:  $d^{n-1}+d-1$ for odd $d$ and $d^{n-1}+d$ for even $d$.}	\label{tabelka-wymiarow}
	\end{table}
\end{center}
	
If any of the vectors in the set is orthogonal to $s+m$ vectors on one of the sites,
then the excluded maximal cardinality gets larger exactly by $m$.
This eliminates certain orthogonality graphs for  GUPBs. Moreover, some graphs can be further excluded in the general case if the $s$ vectors mentioned in the proof are orthogonal to one of the vectors from a UPB  on a site with nonmaximal local dimension, in which case we utilize the weaker condition (\ref{weaker-condition}) instead of the bound (\ref{bound}).

Notice that the proof does not even assume that the set is a UPB; it merely refers to a set of fully product orthogonal vectors. It may be interpreted as an indication of strong a property biproduct extendibility is.  It will be interesting to see in the future if more cardinalities can be excluded without referring to UPBs. Invoking properties of UPBs, on the other hand, would
make things much harder.

The same technique could be used to bound the cardinalities of UPBs, which are not extendible by triproduct, quadriproduct, etc., vectors.  We omit these derivations here.

We conclude this part of the paper with a few observations. First, if GUPBs existed the well-known method of constructing new UPBs from old ones by adding flags \cite{new-from-old}  would not work in their case.
More precisely, suppose we have $M$ UPBs in $\bigotimes_{l=1}^n\cee{d_l}$, $\calB_i=\{\ket{\psi_j^{(i)}}\}_{j=1}^{k_i}$, $i=1,2,\dots,M$. Then the set $\calB=\{\ket{i-1}\ket{\psi_j^{(i)}}, i=1,2,\dots,M, j=1,2\dots, k_i\}$ is a UPB in $\cee{M} \otimes (\bigotimes_{l=1}^n\cee{d_l})$.
 Clearly, this can never be true for GUPBs as there always exists a biproduct vector orthogonal to $\calB$, e.g., $\ket{0}\otimes \ket{\calB_1^{\perp}}$, where $\ket{\calB_1^{\perp}}$ is any vector from the orthocomplement of $\calB_1$. 
  Moreover, the nonexistence of GUPBs in systems  $\bigotimes_{i=1}^n\cee{d_i}$ implies the nonexistence of GUPBs with the same number of elements but in larger ($N$-partite with $N>n$) systems with local dimensions such that their products in groups equal $d_i$'s. 
  For example, if there is no GUPB with $k$ elements in $\calH_{9,3}:=(\cee{9})^{\otimes 3}$ then there is no GUPB with cardinality $k$ in $\calH_{3,6}:=(\cee{3})^{\otimes 6}=(\cee{3}\otimes \cee{3})^{\otimes 3}$.
Regretfully, this fact in combination with our result cannot be used to  eliminate further  cardinalities of GUPBs. This is because the largest cardinality of a GUPB  forbidden by Proposition \ref{main-result} in the fewer partite system is smaller than the minimal size of a GUPB in the larger system. For instance, the largest $k$ for $\calH_{9,3}$ is $120$, while the minimal permissible size of a GUPB in $\calH_{3,6}$ is $245$. On the other hand, it was shown that tensoring would work and a new $n$-partite GUPB could be constructed from two or more $n$-partite GUPBs in this manner \cite{4x4}.

\section{Conclusions}
We have demonstrated a lower bound on the cardinalities of  multipartite unextendible product bases (UPBs), which are at the same time unextendible with biproduct vectors; we called them genuinely unextendible product bases (GUPBs). We have derived a characterization of systems for which this bound eliminates certain (allowed by the theory of UPBs) sizes of GUPBs. In the particularly interesting case of equal local dimensions, our result is universally applicable and shows that maximal GESs (and many smaller ones) cannot be constructed from UPBs.
 This partially solves an open problem posed in \cite{upb-to-ges}. The proof is elementary and only uses the properties of sets of product vectors, which drives us to conjecture that GUPBs might not exist at all or at least in some multipartite systems.  We hope our research stimulates further work on (dis)proving this conjecture.

Our no-go result contributes to a better understanding of the mathematical structure of UPBs. Nonetheless, it also has important practical implications. Two main ones are that it implies that genuinely multiparty entangled bound entangled states with positive partial transposition of certain ranks cannot be constructed from UPBs, and it shows that strongly nonlocal UPBs unextendible by biproducts do not exist  with all permissible cardinalities (if at all). Further consequences are yet to be recognized.

\section{Acknowledgments}
Discussions with R. Augusiak and G. Rajchel-Mieldzioć at the Center for Theoretical Physics of the Polish Academy of Sciences in Warsaw are acknowledged.

\appendix\section{Proof of Proposition \ref{main-result} (continued)}
Here we prove the remaining two statements of Proposition \ref{main-result}.
\begin{proof}
 Condition (\ref{bound}) is nontrivial in a given setup if the bound it gives on $k$ is not smaller than the minimal size of a GUPB allowed theoretically, that is, when the  inequality 
	\begin{equation}\label{ogolnie}
	\dmin+ D/\dmin -m \le  D/\dmax+ \Big\lfloor \frac{D/\dmax-2}{n-1} \Big\rfloor
	\end{equation}
	holds,
	where $m=0$ if both $\dmin$ and $D/\dmin$ are even and $m=1$ otherwise. We can rewrite it as
	\beqn\label{rewritten}
	&&\hspace{-1cm}\dmin-D\left(\frac{1}{\dmax}\frac{n}{n-1} - \frac{1}{\dmin} \right) \non&& \hspace{+1cm}\le m -\Big\{ \frac{D/\dmax-2}{n-1} \Big\}-\frac{2}{n-1},
	\eeqn
	where $\{x\}=x- \lfloor x \rfloor$ is the fractional part of $x$. It immediately follows that a necessary condition for Eq. (\ref{rewritten}) to hold or, equivalently, Eq. (\ref{bound}) to be nontrivial  is that the expression in large parenthesis on the left-hand side (LHS) is positive, i.e.,
	%
$	(n-1)\dmax < n \dmin$,
	%
which is Eq. (\ref{iff-condition}).
	\newline \indent
When all local dimensions are equal, $d_i=:d$, $i=1,2,\dots,n$, in which case Eqs. (\ref{iff-condition})--(\ref{nieparzyste}) are visibly satisfied, we obtain
\begin{equation}\label{}
d-\frac{d^{n-1}}{n-1} \le m - \Big\{ \frac{d^{n-1}-2}{n-1} \Big\}-\frac{2}{n-1}.
\end{equation}
Instead of the above general inequality, we can consider its strongest form with $m=0$ and the fractional part set to be maximal, i.e., $\frac{n-2}{n-1}$
\begin{equation}\label{}
d-\frac{d^{n-1}}{n-1} \le -\frac{n}{n-1},
\end{equation}
which quite obviously is true for any $n$ and $d$. This settles the case of equal local dimensions.
\newline\indent
We now assume that local dimensions are not the same, $\dmax > \dmin$, and Eq. (\ref{iff-condition}) holds. From these two conditions and Eq. (\ref{rewritten}) we may infer that (i) $\dmin \ge 4$ and (ii) $\dmin \ge n$.   To see (i), we simply insert $\dmin=3$ into Eq. (\ref{iff-condition}) and obtain $\dmax < 3+3/(n-1)$. For $n\ge 4$ this gives $\dmax \le 3$, which contradicts $\dmax > \dmin$, while  for $n=3$ we get $\dmax \le 4$, that is, $(d_1,d_2,d_3)=(3,3,4)$ or $(3,4,4)$, which can be shown by direct substitution not to satisfy Eq. (\ref{rewritten}). To show  (ii), which is actually a value added to the proof, we assume the opposite, $n>\dmin$, and consider the following inequalities:
$(n-1)(\dmax-\dmin)\ge n-1 \ge \dmin$, meaning that with the assumptions made, Eq. (\ref{iff-condition}) is violated and it must be $\dmin \ge n$. Case (i) already covers $p=2$ in Eq. (\ref{nieparzyste}).
\newline\indent
We  further assume $\dmin \ge 4$.
\newline\indent
The rest of the proof will be divided into two cases (a) $n \ge 4$ and (b) $n=3$, with the first one being much more straightforward. In both cases the following quantity is of central importance:
\beqn\label{ksi}
\xi &=& D\left(\frac{1}{\dmax}\frac{n}{n-1} - \frac{1}{\dmin} \right)\non
&=& \frac{\bar{D}}{n-1} \left[n \dmin -(n-1)\dmax\right].
\eeqn
Here $\bar{D}=D/\dmin \dmax$.
\newline\indent
Case (a). We have in general
\begin{equation}\label{bound-xi}
\xi \ge \frac{\bar{D}}{n-1} \ge \frac{\dmin^{n-2}}{n-1}
\end{equation}
meaning that the LHS of Eq. (\ref{rewritten}) is bounded as
\begin{equation}\label{}
\dmin-\xi \le \dmin - \frac{\dmin^{n-2}}{n-1} \le -\frac{4}{3},
\end{equation}
where the second inequality follows from the fact that the function in the middle is nonincreasing  in both $n$ and $\dmin$  and thus its largest value is achieved for $n=4$ and $\dmin=4$. On the other hand, the minimal value of the right-hand side (RHS) of Eq. (\ref{rewritten}) is also $-4/3$, meaning that it is always satisfied for $n\ge 4$.
\newline\indent
Case (b). We first deal with the systems from Eqs. (\ref{parzyste}) and (\ref{nieparzyste}).  Note that $3p-1$ and $3p-2$ correspond to the maximal values of $\dmax$ allowed by the necessary condition (\ref{iff-condition}) in the respective cases.
First, by direct substitution, we readily verify that for systems  $(d_1,d_2,d_3)=(2p,2p,3p-1)$ the inequality (\ref{rewritten}) is indeed violated (note that $m=0$ then). In the case of $(d_1,d_2,d_3)=(2p-1,\tilde{d},3p-2)$, $p\ge 3$ (the case $p=2$ has already been covered above),  we have (now $\bar{D}=\tilde{d}$)
\begin{equation}\label{}
\dmin -\xi= 2p-1-\frac{\tilde{d}}{2} \ge \frac{p}{2} >0.
\end{equation}
This means that even when the LHS of Eq. (\ref{rewritten}) is minimal it will always exceed the RHS, which is at most zero (now $m=1$). 
\newline\indent
We are thus left with two cases of dimensions which need to be verified: (x) $(\dmin,\tilde{d}\ge \dmin,\dmax < \dmax ^{(\mathrm{max})})$, and (xx) $(2p,\tilde{d}>2p, \dmax ^{(\mathrm{max})})$,  where $ \dmax ^{(\mathrm{max})}$ is the largest value of $\dmax$ allowed by (\ref{iff-condition}).
For these cases Eq. (\ref{bound-xi}) is too rough for general conclusions, so we need finer bounds.
\newline\indent
{\it Case (x)}. The following holds:
\begin{equation}\label{}
\dmax^{(max)}= \begin{cases}
\frac{3}{2}\dmin-1 & \mathrm{for\;even} \; \dmin\\
\frac{3}{2}\dmin-\frac{1}{2} & \mathrm{for\;odd} \; \dmin.\\
\end{cases}
\end{equation}
Thus, with the assumption $\dmax < \dmax^{(\mathrm{max})}$, we obtain
\begin{equation}\label{}
\dmin-\xi \le \begin{cases}
\dmin-2\tilde{d} \le - \dmin & \mathrm{for\;even} \; \dmin\\
\dmin-\frac{3}{2}\tilde{d}\le -\frac{1}{2}\dmin & \mathrm{for\;odd} \; \dmin.\\
\end{cases}
\end{equation}
Since the minimal value of the LHS of (\ref{rewritten}) is $-3/2$ ($m=0$ and the maximal fractional part) we conclude that systems $(\dmin,\tilde{d},\dmax < \dmax ^{(\mathrm{\mathrm{max}})})$ satisfy this inequality.
\newline\indent
{\it Case (xx)}. Now
\begin{equation}\label{}
\dmin-\xi=2p-\tilde{d} \le -1
\end{equation}
%
and 
\begin{equation}\label{}
\mathrm{RHS\;of\; (\ref{rewritten})}= m-1 \ge -1
\end{equation}
implying that systems   $(2p,\tilde{d}>2p, \dmax ^{(max)})$ satisfy (\ref{rewritten}).
\newline\indent
This ends the proof.
\end{proof}


\begin{thebibliography}{99}
\bibitem{upb-comment} It makes sense to consider nonorthogonal or generalized UPBs, that is, UPBs with the orthogonality condition dropped (see, e.g., \cite{leeinas, skowronek} and, in particular, \cite{upb-to-ges,universal-ges} in the context of the present paper). However, their mathematical structure is much less elegant and their overall significance as a universal tool appears to be very limited compared to their orthogonal counterparts. In fact,  this distinction is mostly disregarded for this reason in the literature and UPBs are simply assumed to be orthogonal.
	For conciseness,  we avoid a discussion about this issue in the main text.
\end{thebibliography}
\end{document}